\documentclass[a4paper,11pt]{article}

\usepackage{fullpage}
\usepackage{amsmath}
\usepackage{amssymb}
\usepackage{amsthm}
\usepackage{graphicx}
\usepackage{enumerate}
\usepackage{cite}

\newtheoremstyle{mytheorem}{3pt}{3pt}{\slshape}{}{\bfseries}{}{.5em}{}
\theoremstyle{mytheorem}
\newtheorem{theorem}{Theorem}

\newtheorem{lemma}{Lemma}

\theoremstyle{definition}

\newcommand{\Real}{\ensuremath{\mathbb{R}}}
\newcommand{\Plane}{\ensuremath{\mathbb{R}^2}}

\newcommand{\A}{\ensuremath{\mathcal{A}}}
\newcommand{\C}{\ensuremath{\mathcal{C}}}
\newcommand{\D}{\ensuremath{\mathcal{D}}}
\newcommand{\M}{\ensuremath{\mathcal{M}}}
\newcommand{\UE}{\ensuremath{\mathcal{U}}}
\newcommand{\LE}{\ensuremath{\mathcal{L}}}

\newcommand{\Poly}{\ensuremath{\mathcal{P}}}

\newcommand{\bd}{\ensuremath{\partial}}

\newcommand{\seg}{\overline}

\title{Querying Two Boundary Points for Shortest Paths in a Polygonal Domain\thanks{%
An extended abstract version of this paper will appear in Proceedings of
20th International Symposium on Algorithms and Computation (ISAAC 2009).
}}

\author{%
Sang Won Bae\thanks{Supported by the Brain Korea 21 Project.}
\and Yoshio Okamoto\thanks{Supported by Global COE Program
``Computationism as a Foundation for the Sciences'' and Grant-in-Aid
for Scientific Research from Ministry of Education, Science and
Culture, Japan, and Japan Society for the Promotion of Science.} }

\date{
{\small
\footnotemark[2] Department of Computer Science and Engineering, POSTECH, Pohang, Korea.\\
E-mail: \texttt{swbae@postech.ac.kr}\\
\footnotemark[3] Graduate School of Information Science and Engineering,\\ Tokyo Institute of Technology, Tokyo, Japan.\\
E-mail: \texttt{okamoto@is.titech.ac.jp}\\[5mm]}
}

\makeatletter
\newbox\ProofSym
\setbox\ProofSym=\hbox{%
\unitlength=0.18ex%
\begin{picture}(10,10)
\put(0,0){\framebox(9,9){}}
\put(0,3){\framebox(6,6){}}
\end{picture}}
\renewenvironment{proof}[1][Proof.]{\O@proof{#1}}{\O@endproof}
\def\O@proof#1{\trivlist
   \@topsep\z@\@topsepadd\smallskipamount%
   \@ifstar{\item[]}{\item[\hskip\labelsep\it #1 ]}}
\def\O@endproof{\hfill\copy\ProofSym\linebreak\endtrivlist}
\makeatother

\begin{document}

\maketitle

\begin{abstract}
 We consider a variant of two-point Euclidean shortest path query problem:
 given a polygonal domain, build a data structure for two-point shortest path
 query, provided that query points always lie on the boundary of the domain.
 As a main result, we show that a logarithmic-time query for shortest paths
 between boundary points can be performed using
 $\tilde{O}(n^5)$
 preprocessing time and
 $\tilde{O}(n^5)$ space
 where $n$ is the  number of corners of the polygonal domain
 and the $\tilde{O}$-notation suppresses the polylogarithmic factor.
 This is realized by observing a connection between Davenport-Schinzel
 sequences and our problem in the parameterized space.
 We also provide a tradeoff between space and query time; a sublinear time
 query is possible using $O(n^{3+\epsilon})$ space.
 Our approach also extends to the case where query points should lie on
 a given set of line segments.
%
\end{abstract}

\section{Introduction} \label{sec:intro}
A polygonal domains $\Poly$ with $n$ corners and $h$ holes is a
polygonal region of genus $h$ whose boundary consists of $n$ line
segments. The holes and the outer boundary of $\Poly$ are regarded
as \emph{obstacles}. Then, the geodesic distance between any two
points $p,q$ in a given polygonal domain $\Poly$ is defined to be
the length of a shortest obstacle-avoiding path between $p$ and $q$.

The Euclidean shortest path problem in a polygonal domain has drawn
much attention in the history of computational
geometry~\cite{m-spn-04}. In the \emph{two-point shortest path
query} problem, we preprocess $\Poly$ so that we can determine a
shortest path (or its length) quickly for a given pair of query
points $p,q \in \Poly$. While we can compute a shortest path in
$O(n\log n)$ time from scratch~\cite{hs-oaespp-99}, known structures
for logarithmic time query require significantly large
storage~\cite{cm-tpespqp-99}. Chiang and
Mitchell~\cite{cm-tpespqp-99} developed several data structures that
can answer a two-point query quickly with tradeoffs between storage
usage and query time. Most notably, $O(\log n)$ query time can be
achieved by using $O(n^{11})$ space and preprocessing time;
sublinear query time by $O(n^{5+\epsilon})$ space and preprocessing
time.
More recently, Guo et al.~\cite{gms-spqpd-08} have shown that
a data structure of size $O(n^2)$ can be constructed in $O(n^2 \log n)$ time
to answer the query in $O(h \log n)$ time, where $h$ is the number of holes.
Their results are summarized in \tablename\
\ref{table:summary}. For more results on shortest paths in a
polygonal domain, we refer to survey articles by Mitchell~\cite{m-spn-04, m-spaop-96}.

In this paper, we focus on a variant of the problem, in which
possible query points are restricted to a subset of $\Poly$;
the boundary $\bd \Poly$ of the domain $\Poly$ or a set of line segments within $\Poly$.
In many applications, possible pairs of source and destination do not span
the whole domain $\Poly$ but a specified subset of $\Poly$.
For example, in an urban planning problem, the obstacles correspond to
the residential areas and the free space corresponds to the walking
corridors. Then, the query points are restricted to the spots where
people depart and arrive, which are on the boundary of obstacles.

Therefore, our goal is to design a data structure using much less
resources than structures of Chiang and Mitchell \cite{cm-tpespqp-99}
when the query domain is restricted to the boundary of a given
polygonal domain $\Poly$ or to a set of segments in $\Poly$.
To our best knowledge, no prior work seems to investigate this variation.
As a main result, in Section~\ref{sec:log_query}, we present a data structure
of size $O(n^4\lambda_{66}(n))$ that can be constructed
in $O(n^4 \lambda_{65}(n)\log n)$ time and can answer
a $\bd\Poly$-restricted two-point shortest path query in $O(\log n)$ time.
Here, $\lambda_{m}(n)$ stands for the maximum length of
a Davenport-Schinzel sequence of order $m$ on $n$ symbols~\cite{sa-dsstga-95}.
It is good to note that $\lambda_{m}(n)=O(n\log^* n)$ for any constant $m$ as a
convenient intuition,
while tighter bounds are known \cite{sa-dsstga-95,n-ibntdsstg-09}.
We also provide a tradeoff between space and query time in Section~\ref{sec:tradeoff}.
In particular, we show that one can achieve sublinear query time using $O(n^{3+\epsilon})$
space and preprocessing time. New results in this paper are also summarized in \tablename\ \ref{table:summary}.

\begin{table}[t]
\label{table:summary} \centering \caption{Summary of new and known
results on exact two-point shortest path queries, where $\epsilon > 0$ is arbitrary and
$0 < \delta \leq 1$ is a parameter. [new] denotes our results.} 
\begin{tabular}{|c|c|c|c|c|}
\hline
Query domain & Preprocessing time & Space & Query time & Ref. \\
\hline
$\Poly$ & $O(n^{11})$ & $O(n^{11})$ & $O(\log n)$ & \cite{cm-tpespqp-99} \\
$\Poly$ & $O(n^{10}\log n)$ & $O(n^{10}\log n)$ & $O(\log^2 n)$ & \cite{cm-tpespqp-99} \\
$\Poly$ & $O(n^{5+10\delta+\epsilon})$ & $O(n^{5+10\delta+\epsilon})$ & $O(n^{1-\delta} \log n)$ & \cite{cm-tpespqp-99} \\
$\Poly$ & $O(n^{5})$ & $O(n^{5})$ & $O(\log n+h)$ & \cite{cm-tpespqp-99} \\
$\Poly$ & $O(n+h^{5})$ & $O(n+h^{5})$ & $O(h\log n)$ & \cite{cm-tpespqp-99} \\
$\Poly$ & $O(n^2 \log n)$ & $O(n^{2})$ & $O(h \log n)$ & \cite{gms-spqpd-08} \\
$\bd \Poly$ & $O(n^{4}\lambda_{65}(n)\log n)$ & $O(n^{4}\lambda_{66}(n))$ & $O(\log n)$ & [new]\\
$\bd \Poly$ & $O(n^{3+\delta}\lambda_{65}(n^\delta)\log n)$ & $O(n^{3+\delta}\lambda_{66}(n^\delta))$ & $O(n^{1-\delta} \log n)$ & [new] \\
$m$ segments & $O(m^2n^{3+\delta}\lambda_{65}(n^\delta)\log n)$ & $O(m^2n^{3+\delta}\lambda_{66}(n^\delta))$ & $O(n^{1-\delta} \log (m+n))$ & [new] \\
\hline
\end{tabular}
\end{table}
Our data structure is a subdivision of two-dimensional domain
parameterized in a certain way.
The domain is divided into a number of grid cells in which a set of
constrained shortest paths between query points have the same
structure. Each grid cell is divided according to the projection of the
lower envelope of functions stemming from the constrained shortest
paths. With careful investigation into this lower envelope,
we show the claimed upper bounds.

Also, our approach readily extends to the variant where query points are
restricted to lie on a given segment or a given set of segments in $\Poly$.
We discuss this extension in Section~\ref{sec:extension}.

\subsection{Related Work}
In the case where $\Poly$ is a simple polygon ($h=0$), the two-point shortest
path query can be answered in $O(\log n)$ time after $O(n)$ preprocessing time~\cite{gh-ospqsp-89}.
More references and results on shortest paths in simple polygons can be found
in a survey article by O'Rourke and Suri~\cite{os-p-04}

Before Chaing and Mitchell~\cite{cm-tpespqp-99}, fast two-point shortest path queries
in polygonal domains were considered as a challenge.
Due to this difficulty, many researchers have focused on the \emph{approximate}
two-point shortest path query problem.
Chen~\cite{c-apespp-95} achieved an $O(n \log n)$-sized structure for $(6+\epsilon)$-approximate
shortest path queries in $O(\log n)$ time, and also pointed out that
a method of Clarkson~\cite{c-aaspmp-87} can be applied to answer
$(1+\epsilon)$-approximate shortest path queries in $O(\log n)$ time
using $O(n^2)$ space and $O(n^2  \log n)$ preprocessing time.
Later, Arikati et al.~\cite{accdsz-psaspqop-96} have improved the above results based on planar spanners.

The problem on polyhedral surfaces also have been considered.
Agarwal et al.~\cite{aaos-supa-97} presented a data structure of size $O(n^6m^{1+\delta})$
that answers a two-point shortest path query on a given \emph{convex} polytope
in $O((\sqrt{n}/m^{1/4})\log n)$ time after $O(n^6m^{1+\delta})$ preprocessing time
for any fixed $1\leq m \leq n^2$ and any $\delta >0$.
They also considered the problem where the query points are restricted to lie on the edges
of the polytope, reducing the bounds by a factor of $n$
from the general case.
Recently, Cook IV and Wenk~\cite{cw-sppps-09} presented an improved method using
kinetic Voronoi diagrams.

\section{Preliminaries} \label{sec:pre}
We are given as input a polygonal domain $\Poly$ with $h$ holes and
$n$ corners. More precisely, $\Poly$ consists of an outer simple
polygon in the plane $\Plane$ and a set of $h$ ($\geq 0$) disjoint open
simple polygons inside $P$. As a set, $\Poly$ is the region
contained in its outer polygon \emph{excluding} the holes, also
called the \emph{free space}. The complement of $\Poly$ in the plane
is regarded as \emph{obstacles} so that any feasible path does not
cross the boundary $\bd \Poly$ and lies inside $\Poly$. It is well
known from earlier works that there exists a \emph{shortest
(obstacle-avoiding) path} between any two points $p, q \in
\Poly$~\cite{m-spaop-96}.

Let $V$ be the set of all corners of $\Poly$.  Then any shortest path
from $p\in \Poly$ to $q \in \Poly$ is a simple polygonal path and
can be represented by a sequence of line segments connecting points
in $V\cup \{p,q\}$~\cite{m-spaop-96}. The \emph{length} of a
shortest path is the sum of the Euclidean lengths of its segments.
The geodesic distance, denoted by $d(p,q)$, is the length of a
shortest path between $p$ and $q$. Also, we denote by $|\seg{pq}|$
the Euclidean length of segment $\seg{pq}$.

A \emph{two-point shortest path query} is given as a pair of points
$(p,q)$ with $p,q\in \Poly$ and asks to find a shortest path between
$p$ and $q$. In this paper, we deal with a restriction where the
queried points $p$ and $q$ lie on the boundary $\bd \Poly$.

A \emph{shortest path tree} $SPT(p)$ for a given source point $p\in \Poly$ is a
spanning tree on the corners $V$ plus the source $p$ such that the
unique path to any corner $v\in V$ from the source $p$ in $SPT(p)$
is a shortest path between $p$ and $v$. The combinatorial complexity of $SPT(p)$
for any $p\in \Poly$ is at most linear in $n$. A shortest path map
$SPM(p)$ for the source $p$ is a decomposition of the free space
$\Poly$ into cells in which any point $x$ has a shortest path to $p$
through the same sequence of corners in $V$. Once $SPT(p)$ is
obtained, $SPM(p)$ can be computed as an additively weighted Voronoi
diagram of $V \cup \{p\}$ with weight assigned by the geodesic
distance to $p$~\cite{m-spaop-96}; thus, the combinatorial
complexity of $SPM(p)$ is linear. A cell of $SPM(p)$ containing a
point $q\in \Poly$ has the common last corner $v \in V$ along the
shortest path from $p$ to $q$; we call such a corner $v$ the
\emph{root} of the cell or of $q$ with respect to $p$. An
$O(n \log n)$ time algorithm, using $O(n\log n)$ working space, to
construct $SPT(p)$ and $SPM(p)$ is presented by Hershberger and
Suri~\cite{hs-oaespp-99}.

An \emph{SPT-equivalence decomposition} $\A^{SPT}$ of $\Poly$ is the
subdivision of $\Poly$ into cells in which every point has
topologically equivalent shortest path tree. An $\A^{SPT}$ can be
obtained by overlaying $n$ shortest path maps $SPM(v)$ for every
corner $v\in V$~\cite{cm-tpespqp-99}. Hence, the complexity of
$\A^{SPT}$ is $O(n^4)$. Note that $\A^{SPT} \cap \bd \Poly$ consists
of at most $O(n^2)$ points; they are intersection points between any
edge of $SPM(v)$ for any $v\in V$ and the boundary $\bd \Poly$. We
call those intersection points, including the corners $V$, the
\emph{breakpoints}. The breakpoints induce $O(n^2)$ intervals along
$\bd \Poly$. We shall say that a breakpoint is \emph{induced by
$SPM(v)$} if it is an intersection of an edge of $SPM(v)$ and $\bd
\Poly$.

Given a set $\Gamma$ of algebraic surfaces and surface patches in
$\Real^d$, the \emph{lower envelope} $\mathcal{L}(\Gamma)$ of
$\Gamma$ is the set of pointwise minima of all given surfaces or patches
in the $d$-th coordinate. The \emph{minimization diagram}
$\M(\Gamma)$ of $\Gamma$ is a decomposition of $\Real^{d-1}$ into
faces, which are maximally connected region over which
$\mathcal{L}(\Gamma)$ is attained by the same set of functions. In
particular, when $d=3$, the minimization diagram $\M(\Gamma)$ is
simply a projection of the lower envelope onto the $xy$-plane.
Analogously, we can define the \emph{upper envelope} and the
\emph{maximization diagram}.

As we intensively exploit known algorithms on algebraic surfaces or
surface patches and their lower envelopes, we assume a model of
computation in which several primitive operations dealing with a
constant number of given surfaces can be performed in constant time:
testing if a point lies above, on or below a given surface,
computing the intersection of two or three given surfaces,
projecting down a given surface, and so on. Such a model of
computation has been adopted in many research papers;
see~\cite{s-atublehd-94,aas-cefda-97,sa-dsstga-95,as-ata-00}.

%

\section{Structures for Logarithmic Time Query} \label{sec:log_query}
In this section, we present a data structure that answers a
two-point query restricted on $\bd \Poly$ in $O(\log n)$ time.
To ease discussion, we parameterize the boundary $\bd \Poly$. Since
$\bd \Poly$ is a union of $h+1$ closed curves, it can be done by
parameterizing each curve by arc length and merging them into one.
Thus, we have a bijection $p \colon [0, |\bd \Poly|) \to
\bd \Poly$ that maps a one-dimensional interval into $\bd \Poly$,
where $|\bd \Poly|$ denotes the total lengths of the $h+1$ closed curves
forming $\bd \Poly$.
Conversely, the inverse of $p$ maps each interval along $\bd \Poly$ to
an interval of $[0, |\bd \Poly|)$.


%

A shortest path between two points $p, q \in \Poly$ is either the
segment $\seg{pq}$ or a polygonal chain through corners in $V$.
Thus, unless $d(p,q) = |\seg{pq}|$, the geodesic distance is taken
as the minimum of the following functions $f_{u,v} \colon [0, |\bd
\Poly|) \times [0, |\bd \Poly|) \to \Real$ over all $u,v \in V$,
which are defined as follows:
\begin{displaymath}
 f_{u,v}(s,t) :=
 \begin{cases}
 |\seg{p(s)u}| + d(u,v) + |\seg{v p(t)}| & \text{if $u \in VP(p(s))$ and $v\in VP(p(t))$},\\
 \infty & \text{otherwise},
 \end{cases}
\end{displaymath}
where $VP(x)$, for any point $x\in \Poly$, denotes the
\emph{visibility profile} of $x$, defined as the set of all points
$y \in \Poly$ that are \emph{visible from} $x$; that is, $\seg{xy}$
lies inside $\Poly$. The symbol $\infty$ can be replaced by an upper
bound of $\max_{s,t} d(p(s), p(t))$; for example, the total length $|\bd \Poly|$
of the boundary of the polygonal domain $\Poly$.

Since the case where $p(s)$ is visible from $p(t)$, so the shortest
path between them is just the segment $\seg{p(s)p(t)}$, can be
checked in $O(\log n)$ time using $O(n^2 \log n)$
space~\cite{cm-tpespqp-99}, we assume from now on that $p(s) \notin
VP(p(t))$. Hence, our task is to efficiently compute the lower
envelope of the $O(n^2)$ functions $f_{u,v}$ on a 2-dimensional
domain $\D := [0, |\bd \Poly|) \times [0, |\bd \Poly|)$.

\subsection{Simple lifting to 3-dimension}
Using known results on the lower envelope of the algebraic surfaces
in $3$-dimension, we can show that a data structure of size
$O(n^{6+\epsilon})$ for $O(\log n)$ query can be built in
$O(n^{6+\epsilon})$ time as follows.

Fix a pair of intervals $I_s$ and $I_t$ induced by the breakpoints.
Since $I_s$ belongs to a cell of an SPT-equivalence
decomposition, $VP(p(s))$ is independent of choices over all
$s \in I_s$.  Therefore, the set $V_s := V \cap VP(p(s))$ of corners
visible from $p(s)$ is also independent of the choice of $s\in I_s$ and further, for a fixed
$u\in V_s$, there exists a unique $v \in V$ that minimizes $f_{u,v}(s,t)$ for any
$(s,t) \in I_s\times I_t$ over all $v\in V$~\cite{cm-tpespqp-99}.
This implies that for
each such subdomain $I_s \times I_t \subset \D$ we extract at most
$n$ functions, possibly appearing at the lower envelope. Moreover,
in $I_s \times I_t$, such a function is represented explicitly; for
$u \in V_s$ and $v\in V_t$,
 \[ f_{u,v}(s,t) = \sqrt{ (x(s) {-} x_u)^2 + (y(s) {-} y_u)^2} + d(u,v) + \sqrt{
 (x(t){-}x_v)^2 + (y(t){-}y_v)^2},\]
where $x(s)$ and $y(s)$ are the $x$- and the $y$-coordinates of
$p(s)$, and $x_u$ and $y_u$ are the $x$- and the $y$-coordinates of
a point $u \in \Plane$. Note that $x(s)$ and $y(s)$ are linear
functions in $s$ by our parametrization.

For each $u \in V$, there exists a corner $v \in V_t$ that
minimizes $f_{u,v'}$ in $I_s \times I_t$ over all $v' \in V_t$.
Therefore, the function $g_u := \min_{v \in V_t}f_{u,v}$ on $I_s \times I_t$
is well-defined.
Observe that the graph of $g_u$ is an algebraic surface with degree at most $4$ in
$3$-dimensional space. Applying any efficient algorithm that
computes the lower envelope of algebraic surfaces in $\Real^3$, we
can compute the lower envelope of the functions $g_u$ in
$O(n^{2+\epsilon})$ time~\cite{s-atublehd-94}. Repeating this for
every such subdomain $I_s\times I_t$ yields $O(n^{6+\epsilon})$
space and preprocessing time.

Since we would like to provide a point location structure in domain
$\D$, we need to find the minimization diagram $\M$ of the computed
lower envelope. Fortunately, our domain is $2$-dimensional, so we
can easily project it down on $\D$ and build a point
location structure with an additional logarithmic factor.

In another way around, one could try to deal with \emph{surface
patches} on the whole domain $\D$. Consider a fixed corner $u\in V$
and its shortest path map $SPM(u)$. The number of breakpoints
induced by $SPM(u)$ is at most $O(n)$, including the corners $V$
themselves. This implies at most an $O(n^2)$ number of
combinatorially different paths between any two boundary points
$p(s)$ and $p(t)$ via $u$. That is, for a pair of intervals $I_s$
and $I_t$, we have a unique path via $u$ and its length is
represented by a partial function of $(s,t)$ defined on a
rectangular subdomain $I_s \times I_t \subset \D$. Hence, we have
$O(n^2)$ such partial functions for each $u\in V$, and thus $O(n^3)$
in total. Each of them defines an algebraic surface patch of
constant degree on a rectangular subdomain. Consequently, we can
apply the same algorithm as above to compute the lower envelope of those
patches in $O((n^3)^{2+\epsilon}) = O(n^{6+\epsilon})$ time and
space.

\subsection{$O(n^{5+\epsilon})$-space structure}
Now, we present a way of proper grouping of subdomains to reduce the
time/space bound by a factor of $n$.
We call a subdomain $I_s \times I_t \subset \D$, where both $I_s$
and $I_t$ are intervals induced by breakpoints, a \emph{grid cell}.
Thus, $\D$ consists of $O(n^4)$ grid cells. We will decompose $\D$
into $O(n^3)$ \emph{blocks} of $O(n)$ grid cells.

Consider a pair of boundary edges $S, T \subset \bd \Poly$ and let
$b_S$ and $b_T$ be the number of breakpoints on $S$ and on $T$,
respectively. Let $p(s_0), \ldots, p(s_{b_S})$
and $p(t_0), \ldots, p(t_{b_T})$ be the breakpoints on $S$ and on $T$,
respectively, in order $s_0 < s_1 < \cdots < s_{b_S}$ and $t_0 < t_1
< \cdots < t_{b_T}$. Take $b_T$ grid cells with $s \in [s_0, s_1)$
and let $\C := [s_0,s_1) \times [t_0, t_{b_T}) \subseteq [s_0,
s_{b_S}) \times [t_0, t_{b_T})$ be their union. We redefine the
functions $f_{u,v}$ on domain $\C$. As discussed above, for any
$s\in [s_0, s_1)$, we have a common subset $V_s$ of corners visible
from $p(s)$.

For any $u\in V_s$, let $g_u(s,t) := \min_{v\in V} f_{u,v}(s,t)$ be
a function defined on $\C$ and
$b_{T}^u$ be the number of breakpoints on $T$ induced by $SPM(u)$.
The following is our key observation.
\begin{lemma} \label{lemma:patch}
 The graph of $g_u(s,t)$ on $\C$ consists of at most $b_T^u +
 1$ algebraic surface patches with constant maximum degree.
\end{lemma}
\begin{proof}
If $g_u(s,t) = f_{u,v}(s,t)$ for any $(s,t)\in \C$ and some $v\in
V$, then $p(t)$ lies in a cell of $SPM(u)$ with root $v$; by the
definition of $g_u$, the involved path goes directly from $p(s)$ to
$u$ and follows a shortest path from $u$ to $p(t)$. On the other
hand, when we walk along $T$ as $t$ increases from $t_1$ to
$t_{b_T}$, we encounter $b_T^u$ breakpoints induced by $SPM(u)$;
thus, $b_T^u + 1$ cells of $SPM(u)$. Hence, the lemma is shown.
\end{proof}
Moreover, the partial function corresponding to each patch of
$\gamma_u$ is defined on a rectangular subdomain $[s_0, s_1) \times
[t_i, t_j)$ for some $1\leq i < j \leq b_T$. This implies that the
lower envelope of $g_u$ on $\C$ is represented by that of at most
$\sum_u (b_T^u + 1) = n + b_T$ surface patches.

Though this envelope can be computed in $O((n+b_T)^{2+\epsilon})$
time, we do further decompose $\C$ into $\lceil \frac{b_T}{n}
\rceil$ blocks of at most $n$ grid cells. This can be simply done by
cutting $\C$ at $t = t_{in}$ for each $i = 1, \ldots, \lfloor
\frac{b_T}{n} \rfloor$. For each such block of grid cells, we have
at most $2n$ surface patches and thus their lower envelope can be
computed in $O(n^{2+\epsilon})$ time. Hence, we obtain the following
consequence.
\begin{theorem}
 One can preprocess a given polygonal domain $\Poly$ in
 $O(n^{5+\epsilon})$ time into a data structure of size
 $O(n^{5+\epsilon})$ that answers the two-point shortest path
 query restricted to the boundary $\bd \Poly$ in $O(\log n)$-time,
 where $\epsilon$ is an arbitrarily small positive number.
\end{theorem}
\begin{proof}
Recall that $\sum_S b_S = \sum_T b_T = O(n^2)$. For a pair of
boundary edges $S$ and $T$, we can compute the lower envelope of the
functions $f_{u,v}$ in $O(b_S \lceil \frac{b_T}{n} \rceil
n^{2+\epsilon})$. Summing this over every pair of boundary edges, we
have
 \[ \sum_{S,T} O(b_S \lceil \frac{b_T}{n} \rceil n^{2+\epsilon}) =
 O(n^{4+\epsilon})\cdot \sum_T (\frac{b_T}{n} + 1) =
 O(n^{5+\epsilon}).\]
A point location structure on the minimization diagram can be built
with additional logarithmic factor, which is subdued by
$O(n^{\epsilon})$.
\end{proof}

\subsection{Further improvement}
The algorithms we described so far compute the lower envelope of
surface patches in $3$-space in order to obtain the minimization
diagram $\M$ of the functions $f_{u,v}$. In this subsection, we
introduce a way to compute $\M$ rather directly on $\D$, based on
more careful analysis.

Basically, we make use of the same scheme of partitioning the domain
$\D$ into blocks of (at most) $n$ grid cells as in
Section~\ref{sec:log_query}.2. Let $\C$ be such a block defined as
$[s_0, s_1) \times [t_0, t_{b_T})$ such that $[s_0, s_1)$ is an interval
induced by the breakpoints and $[t_0, t_{b_T})$ is a union of
$b_T \le n$ consecutive intervals in which we have $b_T-1$ breakpoints $t_1,
\dots, t_{b_T-1}$.

By Lemma~\ref{lemma:patch}, the functions $g_u=\min_{v\in V}
f_{u,v}$ restricted to $\C$ can be split into at most $2n$ partial
functions $h_i$ with $1\leq i \leq 2n$ defined on a subdomain $\C_i
\subseteq \C$. Each $h_i(s,t)$ is represented explicitly as
$h_i(s,t) = |p(s) u_i| + d(u_i, v_i) + |v_i p(t)|$ in $\C_i$, so
that we have $h_i(s,t) = f_{u_i, v_i}(s,t)$ for any $(s,t) \in \C_i$
and some $u_i, v_i \in V$. Note that it may happen that $u_i = u_j$
or $v_i = v_j$ for some $i$ and $j$; in particular, if $u_i = u_j$,
we have $\C_i \cap \C_j = \emptyset$. Also, as discussed in
Section~\ref{sec:log_query}.2, $\C_i$ is represented as $[s_0,
s_1)\times [t_k, t_{k'})$ for some $0 \leq k < k' \leq b_T$.

In this section, we take the partial functions $h_i$ into account,
and thus the goal is to compute the minimization diagram $\M$ of
surface patches defined by the $h_i$. We start with an ordering on
the set $V_s$ of corners visible from $p(s)$ for any $s\in [s_0,
s_1)$ based on the following observation.
\begin{lemma} \label{lemma:ordering}
 The angular order of corners in $V_s$ seen at $s$ is constant
 if $s$ varies within $[s_0, s_1)$.
\end{lemma}
\begin{proof}
If this is not true, we have such an $s' \in [s_0, s_1)$ that $p(s')$
and two corners $v, v'$ in $V_s$ are collinear. Since corners lie on
the boundary of an obstacle, one of $v$ and $v'$ is not visible
from $p(s)$ locally near $s'$; that is, $p(s')$ is a breakpoint, a
contradiction.
\end{proof}

\begin{figure}[t]
\centering
\includegraphics[width=.4\textwidth]{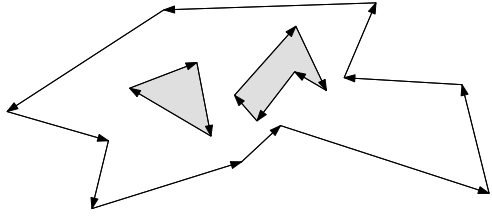}
\caption{As $s$ increases, $p(s)$ moves in direction where the
obstacle is to the right.} \label{fig:orientation}
\end{figure}

Without loss of generality, we assume that as $s$ increases, $p(s)$
moves along $\bd \Poly$ in direction that the obstacle lies to the
\emph{right}; that is, $p(s)$ moves \emph{clockwise} around each
hole and \emph{counter-clockwise} around the outer boundary of
$\Poly$.
(See Figure~\ref{fig:orientation}.)
By Lemma~\ref{lemma:ordering}, we order the corners in $V_s$ in
counter-clockwise order at $p(s)$ for any $s\in [s_0, s_1)$; let
$\prec$ be a total order on $V_s$ such that $u \prec u'$ if and only
if $\angle p(s_0)p(s) u < \angle p(s_0)p(s)u'$.


From now on, we investigate the set
 \[B(i, j) := \{ (s,t)\in \C_i \cap \C_j \mid h_i(s,t) = h_j(s,t) \},\]
which is a projection of the intersection of two surface patches
defined by $h_i$ and $h_j$. One can easily check that $B(i,j)$ is a
subset of an algebraic curve of degree at most $8$.



\begin{lemma} \label{lemma:monotone}
 The set $B(i,j)$ is $t$-monotone. That is, for fixed $t$, there is
 at most one $s\in [s_0, s_1)$ such that $(s,t) \in B(i,j)$.
\end{lemma}
\begin{proof}

If $(s,t) \in B(i,j)$, we have equation $h_i(s,t) = h_j(s,t)$. From
the equation, we get $|p(s)u_i| - |p(s)u_j| = |v_jp(t)| + d(u_j,
v_j) - |v_ip(t)| - d(u_i, v_i)$. If we fix $t$ as constant,
$|p(s)u_i| - |p(s)u_j|$ remains a constant even if $s$ varies within
$[s_0, s_1)$. Thus, $p(s)$ is an intersection point with line
segment $\seg{p(s_0)p(s_1)}$ and a branch $H$ of a hyperbola whose
foci are $u_i$ and $u_j$. If $\seg{p(s_0)p(s_1)} \cap H$ consists of
two points, then the line $\ell$ through $u_i$ and $u_j$ must cross
$\seg{p(s_0)p(s_1)}$ at a point $z$; $\ell$ is the transverse axis
of $H$ and both $u_i$ and $u_j$ lie in one side of the line
supporting $\seg{p(s_0)p(s_1)}$. Such a crossing point $z$ is a
breakpoint by definition but we do not have any breakpoint within
$\seg{p(s_0)p(s_1)}$ since $[s_0, s_1)$ is an interval induced by
the breakpoints, a contradiction.
\end{proof}

\begin{lemma} \label{lemma:increasing}
 The set $B(i,j)$ is either an empty set or an open curve whose endpoints lie on the boundary
 of $\C_i \cap \C_j$. Moreover, $B(i,j)$ is either a linear segment
 parallel to the $t$-axis or $s$-monotone.
\end{lemma}
\begin{proof}
Let $I_s$ and $I_t$ be intervals such that $\C_i \cap \C_j =
I_s\times I_t$. Note that $I_s = [s_0, s_1)$ and $I_t = [t_k,
t_{k'})$ for some $0 \leq k<k' \leq b_T$.

First, note that if $u_i = u_j$, $\C_i \cap
\C_j = \emptyset$ and thus $B(i,j)=\emptyset$, so the lemma is true.
Thus, we assume that $u_i \neq u_j$. Regarding $v_i$ and
$v_j$, there are two cases: $v_i = v_j$ or $v_i\neq v_j$. In the
former case, we get $|p(s)u_i| - |p(s)u_j| = d(u_j, v_j) - d(u_i,
v_i)$ from equation $h_i(s,t) = h_j(s,t)$. Observe that variable $t$
is readily eliminated from the equation, and thus if there exists
$(s',t') \in \C_i \cap \C_j$ with $(s',t')\in B(i,j)$, we have $(s',
t) \in B(i,j)$ for every other $t \in I_t$. Hence, by
Lemma~\ref{lemma:monotone}, $B(i,j)$ is empty or a straight line
segment in $\C_i \cap \C_j$ which is parallel to $t$-axis, and thus
the lemma is shown.

Now, we consider the latter case where $v_i \neq v_j$. Without loss
of generality, we assume that $u_i \prec u_j$. Recall that if $u_i
\prec u_j$, then $\angle p(s_0)p(s) u_i < \angle p(s_0)p(s)u_j$ for
any $s$ in the interior of $I_s$. We denote $\theta_i(s) := \angle
p(s_0)p(s) u_i$ and $\theta_j(s) := \angle p(s_0)p(s) u_j$. On the
other hand, we also have a similar relation for $v_i$ and $v_j$. Let
$\phi_i(t) := \angle p(t_0)p(t)v_i$ and $\phi_j(t) := \angle
p(t_0)p(t)v_j$. Observe that $\phi_i(t)$ and $\phi_j(t)$ are
continuous functions of $t$, and if $\phi_i(t')=\phi_j(t')$ at
$t=t'$, then $p(t')$ is a breakpoint induced by $SPM(u_i)$ or
$SPM(u_j)$. Since $I_t$ contains no such breakpoint induced by
$SPM(u_i)$ or $SPM(u_j)$ in its interior, either $\phi_i(t) <
\phi_j(t)$ or $\phi_i(t) > \phi_j(t)$ for all $t$ in the interior of
$I_t$; that is, the sign of $\phi_j(t) - \phi_i(t)$ is constant.

%

Since for any $s,s' \in I_s$ with $s'>s$ we have $|p(s')p(s)| = s' -
s$ by our parametrization, we can represent $|p(s)u_i| =
\sqrt{(s+a_i)^2+b_i^2}$ and $|v_ip(t)| = \sqrt{(t+c_i)^2+d_i^2}$,
where $a_i$, $b_i$, $c_i$ and $d_i$ are constants depending on
$u_i$, $v_i$, and parametrization $p$. More specifically, $s+a_i$
denotes a signed distance between $p(s)$ and the perpendicular foot
of $u_i$ onto the line supporting $p(I_s)$, and $b_i$ is the
distance between $u_i$ and the line supporting $p(I_s)$. See
Figure~\ref{fig:bisector}. Thus, $h_i(s,t)$ can represented as
$h_i(s,t) = \sqrt{(s+a_i)^2+b_i^2} + \sqrt{(t+c_i)^2+d_i^2} + d(u_i,
v_i)$.

\begin{figure}[t]
\centering
\includegraphics[width=.75\textwidth]{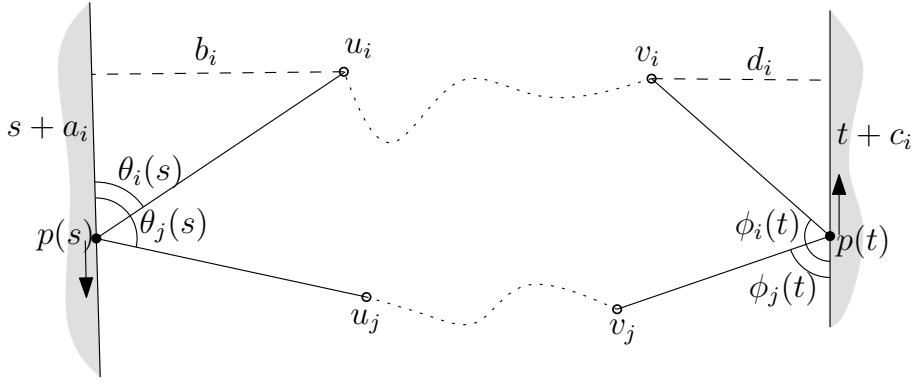}
\caption{Illustration to the proof of Lemma~\ref{lemma:increasing}}
\label{fig:bisector}
\end{figure}
%

Now, we differentiate the both sides of equation $h_i(s,t) =
h_j(s,t)$ by $t$ to obtain the derivative $\frac{d s}{d t}$:
\[\frac{s+a_i}{\sqrt{(s+a_i)^2+b_i^2}} \frac{ds}{dt} +
\frac{t+c_i}{\sqrt{(t+c_i)^2+d_i^2}}  =
\frac{s+a_j}{\sqrt{(s+a_j)^2+b_j^2}} \frac{ds}{dt} +
\frac{t+c_j}{\sqrt{(t+c_j)^2+d_j^2}}.\]
Rearranging this, we obtain
 \[
  \frac{ds}{dt} = \frac {\displaystyle -\frac{t+c_i}{\sqrt{(t+c_i)^2+d_i^2}} + \frac{t+c_j}{\sqrt{(t+c_j)^2+d_j^2}}} {\displaystyle \frac{s+a_i}{\sqrt{(s+a_i)^2+b_i^2}} - \frac{s+a_j}{\sqrt{(s+a_j)^2+b_j^2}}}
  = \frac{-\cos(\phi_i(t))+\cos(\phi_j(t))}{\cos(\theta_i(s))-\cos(\theta_j(s))}.
 \]
Since $0 < \theta_i(s) < \theta_j(s) < \pi$, we have
$\cos(\theta_i(s)) - \cos(\theta_j(s)) > 0$. Also, as discussed
above, $\phi_j(t) - \phi_i(t)$ has a constant sign when $t$ varies
within the interior of $I_t$. Thus, $\cos(\phi_j(t)) -
\cos(\phi_i(t))$ has a constant sign, and $\frac{ds}{dt}$ also has a
constant sign at any $(s,t) \in B(i,j)$. Furthermore,
$\frac{ds}{dt}$ is continuous and has no singularity 
in the
interior of $\C_i\cap\C_j$. This, together with
Lemma~\ref{lemma:monotone}, proves the lemma.
\end{proof}

Now, we know that $B(i,j)$ can be seen as the graph of a partial
function $\{s = \gamma(t)\}$. Also, Lemma~\ref{lemma:increasing}
implies that $B(i,j)$ \emph{bisects} $\C_i \cap \C_j$ into two
connected regions $R(i,j)$ and $R(j,i)$, where $R(i,j) := \{ (s,t)
\in \C_i \cap \C_j \mid h_i(s,t) < h_j(s,t)\}$ and $R(j,i) := \{
(s,t) \in \C_i \cap \C_j \mid h_i(s,t) > h_j(s,t)\}$. Let $\M(i)$ be
the set of points $(s,t)$ where the minimum of $h_j(s,t)$ over all $j$
is attained by $h_i(s,t)$. We then
have $\M(i) = \C_i \setminus \bigcup_j R(j,i)$ for each $i$.

%

For easy explanation, 
from now on, we regard the
$s$-axis as the \emph{vertical axis} in $\D$ so that we can say a
point lies \emph{above} or \emph{below} a curve in this sense.

The idea of computing $\M(i)$ is to use the lower and the upper
envelopes of the bisecting curves $B(i,j)$. In order to do so, we
extend $B(i,j)$ to cover the whole $t$-interval $I_t=[t_k, t_{k'})$ in $\C_i \cap
\C_j$ by following operation: For each endpoint of $B(i,j)$,
if it does not lie on the vertical line $\{ t=t_k\}$ or $\{t=t_{k'}\}$,
attach a horizontal segment to reach the vertical line as shown in Figure~\ref{fig:envelopes}(a).
We denote the resulting curve by $\beta(i,j)$; if
$B(i,j) = \emptyset$, define $\beta(i,j)$ as the horizontal segment
connecting two points $(s_0, t_k)$ and $(s_0, t_{k'})$ in $\C_i \cap \C_j$.
Observe now that $\beta(i,j)$
bisects $\C_i \cap \C_j$ into regions $R(i,j)$ and $R(j,i)$, which
lie \emph{above} and \emph{below} $\beta(i,j)$, respectively.

\begin{figure}[]
\centering
\includegraphics[width=.8\textwidth]{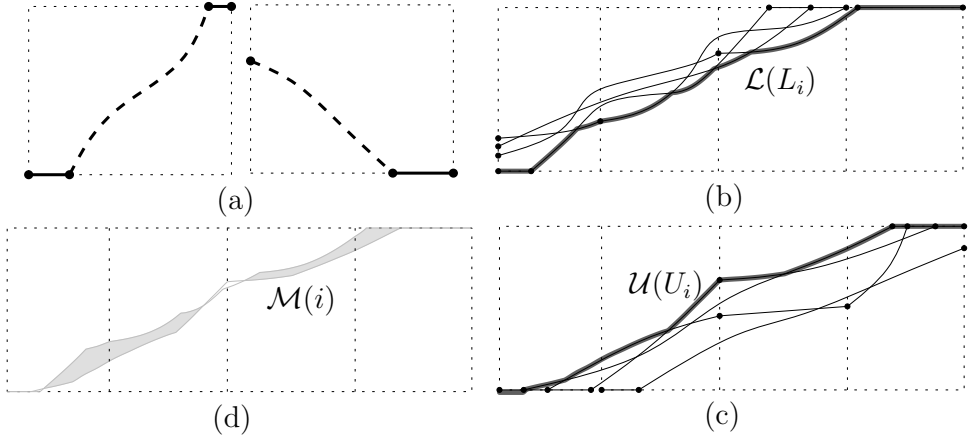}
\caption{(a) How to extend $B(i,j)$ (dashed line) to $\beta(i,j)$ by attaching horizontal segments (solid line);
(b) $\LE(L_i)$, (c) $\UE(U_i)$, and (d) the region $\M(i)$
between them. Here, the dotted boxes are grid cells whose union is
$\C_i$ and the $s$-axis appears vertical.} \label{fig:envelopes}
\end{figure}

Let $\beta(i,j)^+ \subseteq \C_i \cap \C_j$ be the region above
$\beta(i,j)$ and $\beta(i,j)^-$ be the region below $\beta(i,j)$.
For a fixed $i$ with $1\leq i \leq 2n$, we classify the $\beta(i,j)$
into two sets $L_i$ and $U_i$ such that $\beta(i,j) \in L_i$ if
$R(j,i) = \beta(i,j)^+$ or $\beta(i,j) \in U_i$ if $R(j,i) =
\beta(i,j)^-$. Recall that $\M(i) = \C_i \setminus \bigcup_j
R(j,i)$. Thus,
we have
 \[ \M(i) = \C_i \setminus  \left(\bigcup_{\beta \in L_i} \beta^+ \cup \bigcup_{\beta \in U_i} \beta^- \right).\]
The boundary of $\bigcup_{\beta \in L_i} \beta^+$ in $\C_i$ is the lower envelope $\LE(L_i)$
of $L_i$; symmetrically, the boundary of $\bigcup_{\beta \in U_i}
\beta^-$ in $\C_i$ is the upper envelope $\UE(U_i)$ of $U_i$. Therefore,
$\M(i) = \LE(L_i)^- \cap \UE(U_i)^+$, the region below the lower
envelope $\LE(L_i)$ of $L_i$ and above the upper envelope $\UE(U_i)$
of $U_i$, and it can be obtained by computing the overlay of two
envelopes $\LE(L_i)$ and $\UE(U_i)$. See
Figure~\ref{fig:envelopes}(b)--(d). We exploit known results on the
Davenport-Schinzel sequences to obtain the following
lemma~\cite{sa-dsstga-95,h-fuels-89}. 
\begin{lemma} \label{lemma:ds}
 The set $\M(i)$ is of combinatorial complexity $O(\lambda_{66}(n))$ and can be computed in
 $O(\lambda_{65}(n) \log n)$ time, where $\lambda_m(n)$ is the maximum
 length of a Davenport-Schinzel sequence of order $m$ on $n$ symbols.
\end{lemma}
\begin{proof}
$\beta(i,j)$ consists of at most three arcs, at most one algebraic
curve of degree $8$ and at most two straight segments. Thus, we have
at most $6n$ algebraic arcs of degree at most $8$. Any two such arcs
can intersect each other at most $64$ times by B\'{e}zout's Theorem~\cite{h-ag-77}.
Thus, each of $\LE(L_i)$ and $\UE(U_i)$ has complexity
$O(\lambda_{66}(n))$ and can be computed in $O(\lambda_{65}(n) \log
n)$ time~\cite{sa-dsstga-95,h-fuels-89}.

After sorting the vertices on these envelopes in $t$-increasing
order, we can easily specify all intersections between $\LE(L_i)$
and $\UE(U_i)$ in the same bound.
\end{proof}

It should be noted here that the exact constant $66$ is not relevant;
it only matters that this is some constant.

We can compute the minimization diagram $\M$ by computing each
$\M(i)$ in $O(n \lambda_{65}(n) \log n)$ time. In the same time
bound, we can build a point location structure on $\M$. Finally, we
conclude our main theorem.
\begin{theorem} \label{thm:log-query}
 One can preprocess a given polygonal domain $\Poly$ in
 $O(n^4\lambda_{65}(n) \log n)$ time into a data structure of size
 $O(n^4\lambda_{66}(n))$ that answers the two-point shortest path
 query restricted on the boundary $\bd \Poly$ in $O(\log n)$ time.
\end{theorem}

\section{Tradeoffs Between Space and Query Time} \label{sec:tradeoff}
In this section, we provide a space/query-time tradeoff. We use the
technique of partitioning $V$, which has been introduced in Chiang and
Mitchell~\cite{cm-tpespqp-99}.

Let $\delta$ be a positive number with $0 < \delta \leq 1$. We
partition the corner set $V$ into $m = n^{1-\delta}$ subsets $V_1,
\ldots, V_m$ of near equal size $O(n^\delta)$. For each such subset
$V_i$ of corners, we run the algorithm described above with little
modification: We build the shortest path maps $SPM(u)$ only for
$u\in V_i$ and care about only $O(n^{1+\delta})$ breakpoints induced
by such $SPM(u)$. Thus, we consider only the paths from $p(s)$ via
$u \in V_i$ and $v \in V$ to $p(t)$, and thus $O(n^{1+\delta})$
functions $f_{u,v}$ for $u\in V_i$ and $v\in V$.

Since we deal with less number of functions, the cost of
preprocessing reduces from $O(n)$ to $O(n^\delta)$ at several
places. We take blocks of $O(n^\delta)$ grid cells contained in $\D$
and the number of such blocks is $O(n^{2+\delta})$. For each such
block, we spend $O(n^\delta \lambda_{65}(n^\delta) \log n)$ time to
construct a point location structure for the minimization map $\M_i$
of the functions. Iterating all such blocks, we get running time
$O(n^{2+2\delta}\lambda_{65}(n^\delta) \log n)$ for a part $V_i$ of
$V$. Repeating this for all such subsets $V_i$ yields
$O(n^{3+\delta}\lambda_{65}(n^\delta) \log n)$ construction time.

Each query is processed by a series of $m$ point locations on every
$\M_i$, taking $O(m\log n) = O(n^{1-\delta} \log n)$ time.

\begin{theorem} \label{thm:tradeoff}
 Let $\delta$ be a fixed parameter with $0 < \delta \leq 1$.
 Using $O(n^{3+\delta}\lambda_{65}(n^\delta) \log n)$ time and
 $O(n^{3+\delta}\lambda_{66}(n^\delta))$ space,
 one can compute a data structure for $O(n^{1-\delta} \log n)$-time two-point shortest path
 queries restricted on the boundary $\bd \Poly$.
\end{theorem}

Remark that when $\delta = 1$, we obtain
Theorem~\ref{thm:log-query}, and that $O(n^{3+\epsilon})$ time and
space is enough for sublinear time query. Note that if $O(n)$ time
is allowed for processing each query, $O(n^2)$ space
and $O(n^2\log n)$ preprocessing time is sufficient.

\section{Extensions to Segments-Restricted Queries} \label{sec:extension}
Our approach easily extends to the two-point queries in which
queried points are restricted to be on a given segment or a given
polygonal chain lying in the free space $\Poly$.

Let $\mathcal{S}_s$ and $\mathcal{S}_t$ be two sets of $m_s$ and $m_t$ line segments,
respectively, within $\Poly$. In this section, we restrict a query
pair $(p,q)$ of points to lie on $\mathcal{S}_s$ and $\mathcal{S}_t$ each. We will refer
to this type of two-point query as a \emph{$(\mathcal{S}_s, \mathcal{S}_t)$-restricted two-point query}.
As we did above, we take two segments $S \in \mathcal{S}_s$ and $T \in
\mathcal{S}_t$ and let $b_S$ and $b_T$ be the number of breakpoints --- the
intersection points with an edge of $SPM(v)$ for some $v\in V$ ---
on $S$ and $T$, respectively. Also, parameterize $S$ and $T$ as
above so that we have two bijections $p : [0, |S|] \to S$
and $q : [0,|T|] \to T$.

Any path from a point on $S$ leaves to one of the two sides of $S$.
Thus, the idea of handling such a segment within the free space
$\Poly$ is to consider two cases separately. Here, we regard $S$ and
$T$ as directed segments in direction of movement of $p(s)$ and
$q(t)$ as $s$ and $t$ increases, and consider only one case where
paths leave $S$ to its left side and arrive at $T$ from its left
side. The other cases are analogous.

Then, the situation is almost identical to that we considered in
Section~\ref{sec:log_query}.2. For a pair of segments $S$ and $T$,
we can construct a query structure in $O(b_S \lceil \frac{b_T}{n}
\rceil n \lambda_{65}(n)\log n)$ time. Unfortunately, $b_S$ and
$b_T$ can be as large as $O(n^2)$, yielding the same time bound for
$(\bd \Poly, \bd \Poly)$-restricted two-point queries in the worst case.
Thus, in the worst case, we need additional factor of $m_sm_t$ as follows.

\begin{theorem} \label{thm:convexchain}
 Let $\mathcal{S}_s$ and $\mathcal{S}_t$ be two sets of $m_s$ and $m_t$ (possibly crossing)
 line segments, respectively, within $\Poly$, and $\delta$ be a fixed parameter with $0 < \delta \leq 1$.
 Then, using $O(m_sm_tn^{3+\delta}\lambda_{65}(n^\delta) \log n)$ time and
 $O(m_sm_tn^{3+\delta}\lambda_{66}(n^\delta))$ space,
 one can compute a data structure for $O(n^{1-\delta} \log (n+m_s+m_t))$-time
 $(\mathcal{S}_s, \mathcal{S}_t)$-restricted two-point queries.
\end{theorem}

Remark that in practice we expect that the number of breakpoints
$\sum_S b_S$ and $\sum_T b_T$ is not so large as $\Omega(n^2)$ that
the preprocessing and required storage would be much less than the
worst case bounds.

\section{Concluding Remarks} \label{sec:conclusion}
In this paper, we posed the variation of the two-points query
problem in polygonal domains where the query points are restricted
in a specified subset of the free space. And we obtained
significantly better bounds for the boundary-restricted two-point
queries than for the general queries.

Despite of its importance, the two-point shortest path query problem
for the polygonal domains is not well understood. There is a huge
gap ($O(n)$ to $O(n^{11})$) about logarithmic query between the
simple polygon case and the general case but still the reason why we
need such a large storage is still unclear. On the other hand,
restriction on the query domain provides another possibility of
narrowing the gap with several new open problems: %
(1) What is the right upper bound on the complexity of the lower
envelope defined by the functions $f_{u,v}$ on the parameterized
query domain? And what about any lower bound construction? %
(2) If the query domain is a simple $2$-dimensional shape, such as a
triangle, then can one achieve a better performance than the general
results by Chiang and Mitchell?

We would carefully conjecture that our upper bound $\tilde{O}(n^5)$
for logarithmic query could be improved to $\tilde{O}(n^4)$. Indeed, we
have $O(n^4)$ grid cells on the parameterized query domain and
whenever we cross their boundaries, changes in the involved
functions are usually bounded by a constant amount. Thus, if one
could find a clever way of updating the functions and their lower
envelope, it would be possible to achieve an improved bound.

\section*{Acknowledgement}
The authors thank Matias Korman for fruitful discussion.

{ 

}

\end{document}